\documentclass[sigconf]{acmart} 
\AtBeginDocument{%
  \providecommand\BibTeX{{%
    \normalfont B\kern-0.5em{\scshape i\kern-0.25em b}\kern-0.8em\TeX}}}

\usepackage{xcolor}
\usepackage{geometry}

\usepackage{amsmath}
\usepackage[ruled,vlined]{algorithm2e}
\SetKwInput{KwInput}{Input}
\SetKwInput{KwRet}{Return}

\usepackage{graphicx}
\usepackage{pstricks}
\usepackage{inputenc}
\usepackage{hyperref}
\usepackage{url}
\usepackage{paralist}

\usepackage{balance} 

\copyrightyear{2021}
\acmYear{2021}
\acmDOI{}
\acmPrice{}
\acmISBN{}

\acmSubmissionID{461}

\title[Network Fairness and Scalability in Blockchains]{We might walk together, but I run faster: Network Fairness and Scalability in Blockchains}

\author{Anurag Jain}
\email{anurag.jain@research.iiit.ac.in}
\orcid{0000-0002-4637-4708}
\author{Shoeb Siddiqui}
\email{shoeb.siddiqui@research.iiit.ac.in}
\author{Sujit Gujar}
\email{sujit.gujar@iiit.ac.in}
\orcid{0000-0003-4634-7862}
\affiliation{%
  \institution{Machine Learning Lab, \\International Institute of Information Technology}
  \streetaddress{3\textsuperscript{rd} Floor, KCIS Building\\ IIIT-H Campus, Gachibowli}
  \city{Hyderabad}
  \state{Telangana}
  \country{India}
  \postcode{500032}
}

\begin{abstract}
Blockchain-based Distributed Ledgers (DLs) promise to transform the existing financial system by making it truly democratic. In the past decade, blockchain technology has seen many novel applications ranging from the banking industry to real estate. However, in order to be adopted universally, blockchain systems must be scalable to support a high volume of transactions. As we increase the throughput of the DL system, the underlying peer-to-peer network might face multiple levels of challenges to keep up with the requirements. Due to varying network capacities, the slower nodes would be at a relative disadvantage compared to the faster ones, which could negatively impact their revenue. In order to quantify their relative advantage or disadvantage, we introduce two measures of network fairness, $p_f$, the probability of frontrunning and $\alpha_f$, the publishing fairness. We show that as we scale the blockchain, both these measures deteriorate, implying that the slower nodes face a disadvantage at higher throughputs.  It results in the faster nodes getting more than their fair share of the reward while the slower nodes (slow in terms of network quality) get less. Thus, fairness and scalability in blockchain systems do not go hand in hand.
\par In a setting with rational miners, lack of fairness causes miners to deviate from the ``longest chain rule'' or \emph{undercut}, which would reduce the blockchain's resilience against byzantine adversaries. Hence, fairness is not only a desirable property for a blockchain system but also essential for the security of the blockchain and any scalable blockchain protocol proposed must ensure fairness.
\end{abstract}


\newcommand{\BibTeX}{\rm B\kern-.05em{\sc i\kern-.025em b}\kern-.08em\TeX}

\begin{document}


\pagestyle{fancy}
\fancyhead{}


\maketitle 


\section{Introduction}\label{sec:introduction}
Blockchain-based \emph{Distributed Ledgers} (DLs) promise to transform the existing financial system. The idea behind such a transformation is to replace centralized institutions that govern the system by a \emph{decentralized} peer-to-peer network of nodes. The key idea in such a DL system is that the system offers the right incentives to the nodes to act honestly according to the blockchain protocol's rules. Thus, any node can voluntarily choose to participate in the system and incur a computational cost in the expectation of being rewarded. We call such participants in the DL \emph{nodes} or \emph{agents} \footnote{We use nodes and agents interchangeably. We refer to them as nodes when we consider them as a part of a network and agents when we consider them as players in a mining game.}. We believe that such a system can be truly democratic since anyone can choose to participate.

\par If a decentralized system is not \emph{fair}, i.e., the agents do not receive proportionate incentives, they will prefer not to join the system. Consequently, the system will not remain democratic and decentralized if it excludes some agents. In this work, we analyze the fairness characteristics of blockchain-based DLs for which, we consider all the agents being honest but with different network capacity. Although the fairness properties that we describe appear to be healthy for current blockchain systems, they deteriorate quickly as we scale the system to higher throughputs. 

\subsection{Overview}
\par Although cryptocurrencies like Bitcoin and Ethereum are quite popular today, they still lag behind centralized payment systems like Visa in terms of transaction rates and time to finality. As of October 2020, Bitcoin's and Ethereum's network processes an average of 3-4 and 10 transactions per second (TPS), respectively. In contrast, Visa's global payment system handled a reported 1,700 TPS and claimed to be capable of handling more than 24,000 TPS \cite{visa_claim}. For a cryptocurrency to be adopted universally, it must be able to scale to process transactions at much higher throughput, i.e., TPS rate. Hence, blockchain protocols must be scalable to be suitable for widespread adoption.

%
\par However, there are many challenges on the road to scaling block-chain based DLs. Garay et al. \cite{cryptoeprint:2014:765} and Kiyayas et al. \cite{cryptoeprint:2015:1019} show that existing blockchain protocols suffer from a loss of security properties as we scale the system. These security properties are fundamental to the operation of a robust DL. \cite{cryptoeprint:2014:765} and \cite{cryptoeprint:2015:1019} consider a model with two types of agents, honest and adversarial where the adversary tries to attack the ledger by strategically forking the blockchain. A successful fork would allow the adversary to perform a double-spending attack.
%

In this paper, first, we consider a setting in which all agents are honest and show that disparities in the connection to the peer-to-peer network can make the system unfair. In such a case, nodes with a better internet connection will be able to grab a larger share of the reward while those with slower connections might lose out. We show that this disparity significantly increases as we increase the throughput of the system. Notice that improving the quality of the overlay network may be more complicated than making protocol-level changes that may be implemented by merely updating the software clients. 

\par In literature, it is typically assumed that all the agents have equal access to the network, albeit with some finite delay. However, this is seldom the case in practice where some nodes may have better internet connections than others. For the first time, we introduce asymmetry in modeling network connections by assuming different delays for different nodes. Hence, faster nodes would have shorter delays, while slower nodes would have longer delays which in turn results in asymmetry in the rewards collected by these agents. We first analyse consequences of this model in a setting with honest agents and then extend our discussion to rational agents.


\par In order to analyze and quantify network fairness, we introduce two measures of fairness based on network events associated with broadcasting a transaction and broadcasting a block. First, we introduce \emph{frontrunning}, an event associated with a node receiving a transaction. Frontrunning (that we deal with in this paper) occurs when a node confirms a transaction before someone else hears about the transaction. We measure $p_f$, the probability of this event happening between two fractions of the network.
If $p_f$ is high, the faster nodes would consistently be able to grab high-value transactions while the slower ones would only be able to pick low-value ones left out by others. Thus, a high $p_f$ would negatively impact some agents' revenue. We show that if we try to scale a Bitcoin-like system to the throughput offered by the likes of Visa, $p_f$ approaches to nearly 1, which implies that the slower nodes in the system will rarely be able to mine any high-value transactions that would result in these nodes receiving minimal reward in exchange for their mining efforts.

\par We then consider the process of broadcasting a block through the network. \emph{Publishing fairness} quantifies the advantage a node might have over other nodes in broadcasting a block. If a node is able to propagate its block faster than other nodes, in case of an eventual fork, its fork would have a higher probability of being accepted. Since we know that at higher throughputs forks become more common, faster nodes would be able to get more blocks accepted while those of slower nodes would frequently be orphaned. Thus, the slower nodes, would not be able to even gather the fixed block rewards. 
\par As both of these measures deteriorate as with increased throughput, small variations in network access may lead to the system becoming unfair for the slower nodes. This would result in some agents gaining more than their fair share of reward while some agents earn less. This could certainly impact the profitability of the agents that earn less since they still need to pay for the costs associated with mining. Thus, it may lead to drop in the agents maintaining the DL since agents that are unable to accumulate enough reward to break even the mining costs might shut down their mining operation or they might adopt strategic behavior to collect more rewards than that obtained by following the protocol honestly, either of which would reduce the security of the blockchain.

\par We discuss possible behavior that a lack of network fairness could elicit from rational agents. Their behavior could potentially hurt the stability of the system and reduce the effective throughput of the system. We use simulations to show that as the fairness reduces, the default strategy mining on top of the longest chain does not remain the dominant strategy which means that rational agents gain more reward by intentionally forking the longest chain. This could have adverse effect on the resilience of the blockchain against byzantine adversaries, making it less secure\footnote{A byzantine adversary typically tries to defraud the users of the payment system by trying to create double-spending transactions.}.
Hence, even though we scale the system to increase the throughput, we might not find much practical advantage due to these issues. 
\par Thus, the potential of blockchain technology is hindered by the capabilities of the underlying networking infrastructure.\\

Hence, in this work, we:
\begin{enumerate}
    \item Introduce a notion of network fairness in blockchains and highlight its importance. (Section \ref{ssec:fairness_dfns})
    In particular, we introduce \emph{frontrunning} (Definition \ref{def:frontrun}) and \emph{publishing fairness} (Definition \ref{def:rel_fairness}) in the context of network fairness and analyze them for existing blockchain systems.
    \item Study the effect of increasing throughput on network fairness. (Section \ref{sec:analysis}) and provide bounds on \emph{frontrunning} as well as  \emph{publishing fairness} (Theorems \ref{thm:frontrun} and \ref{thm:alpha_f})
    \item Discuss a few consequences of a lack of fairness in terms of creating strategic deviations that may be detrimental to the security of the blockchain. (Section \ref{sec:strategic})
\end{enumerate}
\section{Preliminaries} \label{sec:preliminaries}
A distributed ledger (DL) is a database replicated and shared across multiple nodes in consensus. Blockchain is a type of append-only ledger. When a node wants to append a transaction to the ledger, it broadcasts it to other nodes. All nodes vote by a consensus algorithm to invalidate the existing ledger and replace it with the updated one. (In case of Bitcoin, once a new block is added to the longest chain, the original chain is not considered valid in the presence of the longer chain \cite{nakomoto})
\subsection{Proof-of-Work}
\par Nodes in a proof-of-work blockchain system vote on changes to the DL via their CPU (or GPU in some cases) by trying to mine a block. The chain having the most blocks and correspondingly the maximum proof-of-work is selected as the consensus value. In order to mine a block, the player must successfully find a ``nonce'' value that, along with the other contents of a block, hashes to a value less than a given target. Due to the nature of the hash function, this is a random process, and miners must repeatedly sample different nonce values to mine a block successfully.
\begin{definition}[Fail Function]
We define $\text{fail}(\phi, t)$ as the probability of $\phi$ fraction of the network failing to mine a block in $t$ units of time.
\end{definition}
$\text{fail}(\phi, t) = (1-p)^{\phi\times Ht}$ where $p$ is the probability of a query being successful and $H$ is the cummulative hash rate of the entire network. A more detailed explanation for this is provided in Section \ref{proof:fail} of the supplementary material.
\\For convenience, we sometimes use $\text{fail}(\phi) = \text{fail}(\phi, 1)$, the probability of $\phi$ fraction of the network failing to mine a block in a unit of time.
\subsection{Throughput of a Blockchain-based DL} \label{ssec:speed_security_tradeoff}
The throughput of the blockchain-based DL system depends on two parameters:
\begin{inparaenum}[(i)]
    \item the Block Creation Rate $\lambda$ and 
    \item the Block Size $b$
\end{inparaenum}
with the throughput being proportional to $\lambda \times b$.
\par Although we can increase any of the two parameters to increase the throughput, we suffer consequences due to the limitations of the underlying peer-to-peer network.
\paragraph{Effect of increasing the block creation rate}
By increasing the block creation rate, we risk a node mining a block before it receives the latest block mined by the network.
\paragraph{Effect of increasing the block size}
According to observations by Decker and Wattenhofer \cite{decker2013information}, there exists a linear relationship between the block size b and the time taken by the block to be propagated throughout the network. Hence, if we increase the block size, the time taken by it to propagate increases.
\par Increasing the block size or increasing the block creation rate lead to the deterioration of security properties in Bitcoin-like linear blockchain protocols, as shown by Sompolinsky and Zohar \cite{cryptoeprint:2013:881}. Our contribution is parallel to them since we show that they also cause a loss of fairness.  In this paper, we consider settings with honest players and settings with rational players independently, we also define adversarial players here to show that they can pose threat to security of the blockchain.
\begin{definition}[Honest Agent]
An agent is said to be an \emph{honest player} if and only if it does not deviate from the protocol rules.
\end{definition}
\begin{definition}[Rational Agent]
An agent is said to be a \emph{rational player} if it may strategically deviate from the protocol rules if the deviation is expected to yield a higher reward.
\end{definition}
\begin{definition}[Adversarial Agent]
An agent is said to be an \emph{adversarial player} (and sometimes also known as ``Byzantine'' in the literature) may also strategically deviate from the protocol. However, the adversary's goal is to disrupt the operation of the protocol, and it does \textbf{not} try to maximize its reward.
\end{definition}
\subsection{Reward Model} \label{ssec:reward_model}
There are two principal ways of rewarding the nodes in a blockchain system:
\begin{enumerate}
    \item \emph{Block Reward} - The reward miners can assign to themselves for mining a block. This reward mints new currency, adding to the total amount of currency in circulation.
    \item \emph{Transaction Fees} - This is the fee offered by the users for miners' services by providing incentives to include their transactions in the blocks. Typically, the users are allowed to decide the transaction fees they wish to offer while creating a transaction.
\end{enumerate}
\subsubsection{Other sources of Reward}
Although we consider only block rewards and transaction fees to be the contributors of reward to the miners, miners might earn additional revenue from implicit sources as well. For instance, Daian et al. \cite{daian2019flash} show that there is considerable miner value in \emph{Order Optimization} in \emph{Decentralized Exchanges}, where the miners can rearrange transactions and potentially insert their own in a block to yield a higher reward. In this case, the miners that frontrun can gain additional revenue by grabbing revenue from \emph{order optimization}.
\subsection{Network Model} \label{ssec:network_model}
Typically, the nodes participating in the blockchain communicate with each other using the internet. They form a peer-to-peer network where each node is connected to a few other nodes (which we refer to as neighbors). We assume that the communication between two nodes has a finite delay. This communication delay may vary significantly depending upon numerous factors such as network congestion, network outages, and ISP bottlenecks, but for the sake of analysis, we abstract out these factors.\\
In this paper, we are concerned with the following two actions over the network:
\begin{enumerate}
    \item Broadcasting a Transaction: when a user wishes to make a transaction on the ledger, he/she cryptographically signs a transaction and sends it to a small number of nodes. Each node in the network propagates the transaction to its neighbors. Therefore, a transaction reaches all the nodes in the network after some time.
    \item Broadcasting a Block: Similarly, when a node mines a new block, it sends it to its neighbors, who then propagate it further. Therefore, a block reaches all the nodes in the network after some time.
\end{enumerate}
There are some delays associated with both processes. We assume the delay of receiving a transaction and broadcasting a block to be dependent upon both the quality of a node's connection and the quality of the overall peer-to-peer network. We quantify this delay as $\delta$ refers to the total time taken for any broadcast by a node to reach all nodes in the network.
\section{Different Notions of Fairness}
\label{ssec:fairness_dfns}
\subsection{$\eta$-approximate fairness}
Pass and Shi \cite{pass2017fruitchains} defined $\eta$-approximate fairness as follows
\footnote{In their original paper, Pass et al. used the term $\delta$-approximate fairness however we use a different symbol to distinguish it from the $\delta$ delay we use throughout the paper}:
\begin{definition}[$\eta$-approximate fairness]{\cite{pass2017fruitchains}}
\label{def:approx_fairness}
A blockchain protocol has $\eta$-approximate fairness if, with overwhelming probability, any honest subset controlling $\phi$ fraction of the compute power is guaranteed to get at least a $(1-\eta)\phi$ fraction of the blocks in a sufficiently long window.
\end{definition}

\par The intuition behind this is that in a fair protocol, a agent that controls a $\phi$ fraction of the computational resources should receive a $\phi$ fraction of the rewards. Though $\eta$-approximate fairness has its own merit and importance, the definition manages to capture the intuition only if the reward for each block is similar. In DLs that operate at a slower speed, this variation may not be significant. However, at higher throughputs, the block reward may vary significantly among the blocks. In which case, we should also factor in the reward that the agents get for their blocks. Secondly, this definition does not capture the disparity among different nodes, i.e., which nodes gain more than their fair share and which nodes get less. The measures that we define compare different nodes or sets of nodes to highlight which ones are at a relative advantage and which ones are at a disadvantage.
\par Here, we wish to analyze network fairness and establish measures independent of the computational power of the nodes we are comparing. Hence, we base our definitions on network events instead.

\subsection{Frontrunning ($p_f$)}\label{def:frontrun}

In some cases, it might be easier to analyze and quantify ``unfairness'' rather than fairness. Since only the transaction that comes first is said to be the valid one, any subsequent blocks that include a copy of the transaction lose out on the transaction fees and waste their space, which could have accommodated an unconfirmed transaction. For every transaction, we can imagine a race among the nodes to grab its transaction fees by mining a block that includes it. We can say that the node that manages to win the race by mining a block containing the transaction before everyone else wins the race and has successfully \emph{frontrun} everyone else\footnote{As we discuss in later sections, this may not be the only way to win the race. In some cases agent can change the results of the race by \emph{strategically} deviating from the protocol.}\footnote{We borrow the term from Wall Street jargon where the term originates from the era when stock market trades were executed via paper carried by hand between trading desks. The routine business of hand-carrying client orders between desks would normally proceed at a walking pace. However, a broker could run in front of the walking traffic to reach the desk and execute his order first. \cite{wiki:frontrunning}}. The system will be fair if every node has an equal probability of winning the race. 

\par We define the event \texttt{frontrun\_1} for a node $L$ ($L$ here, stands for Loser) as the event in which he loses the race described earlier. In this case, the node does not gain any reward since he failed to mine the block before everyone else. \\
We define the event \texttt{frontrun\_2} for $L$ as the event in which some other node manages to win the race even before $L$ starts the race. That is, some other node mined a block containing the transaction before $L$ receives the transaction. 
Clearly, $\texttt{frontrun\_2} \subseteq \texttt{frontrun\_1}$. \\

To capture it more formally, 
\begin{definition}[$p_f$]
We call $\{p_f\}^{M}_{m}$ the probability of the event \texttt{frontrun\_2} between the top $M$ percentile of the nodes and the bottom $1-m$ percentile of the nodes in terms of network delays. That is the probability that some node in the top $M$ percentile (in terms of network speed) manages to \texttt{frontrun\_2} all nodes in the bottom $1-m$ percentile.
\end{definition}

Ideally, $P(\texttt{frontrun\_2}) = 0$ since all nodes should start at the same time (this being a race). However, as we analyze in Section \ref{ssec:frontrunning_analysis}, this probability may become significantly high due to varying network speeds.

\subsection{Publishing Fairness ($\alpha_f$)}\label{def:rel_fairness}

Faster internet speed can not only provide an advantage in receiving new transactions but also yield an upper hand in broadcasting blocks. Consider two nodes A and B that mine a block simultaneously. We define $\alpha_f$ as the ratio of the probability that the majority accepts A's block to the probability that the majority accepts B's block. Thus, it quantifies the advantage of A in terms of publishing a block and claiming the associated reward. A lack of publishing fairness implies that not only slower nodes are less likely to receive reward transaction fees in the mined block but the are also less likely to receive the fixed block reward associated with mining a new block.
The intuition behind this definition is that over multiple rounds, this would be the ratio of their conflicting blocks getting accepted.

\[\alpha_f(A,B) = \frac{P(\frac{\text{A's block getting accepted}}{\text{A and B mine a block simultaneously}})}{P(\frac{\text{B's block getting accepted}}{\text{A and B mine a block simultaneously}})}\]

\noindent\textbf{Remark 1}\label{remark:alpha_f_unfair}
Like $p_f$, $\alpha_f$ is also a measure of ``unfairness'' rather than fairness, i.e., higher $\alpha_f$ implies lesser fairness in the system. The ideal value of $\alpha_f$ is 1, when the blocks mined simultaneously by both the nodes are equally likely to get accepted.

\noindent\textbf{Remark 2}\label{remark:alpha_f_simultaneous}
$\alpha_f$ only accounts for the events in which both A and B mine their blocks simultaneously since both the probabilities are conditioned on the blocks being mined simultaneously.

\noindent\textbf{Remark 3}\label{remark:alpha_f_rewards}
Ideally $\alpha_f$ should be close to 1. The high value of $\alpha_f$ indicates that the faster agent can collect more block rewards than the slower agent even though they both do exert the same amount of computation. 

\setlength{\textfloatsep}{0.1cm}
\setlength{\floatsep}{0.1cm}
\section{Analyzing Network Fairness\label{sec:analysis}}
A node that can hear new transactions and propagate blocks faster than other nodes gains an unfair advantage over its peers. In this section, we analyze and quantify the advantage. We believe that the results presented in \ref{par:frontrun_result} and \ref{par:a_f_results} to be a significant insight offered by our paper.

\subsection{Analyzing Frontrunning}
\label{ssec:frontrunning_analysis}
Consider a node with a poor network connection that receives a transaction with a significant delay as compared to others that do not. We now analyze the probability of the event $\texttt{frontrun\_2}$ happening for the node. For the Bitcoin Network, according to \cite{bitcoinstats} it takes less than 4 seconds for a transaction to reach the 50\textsuperscript{th} percentile but more than 15 seconds to reach the 90\textsuperscript{th} percentile. It is likely that the slower nodes only hear about a transaction once a faster node has already confirmed it. Let us consider the probability of $50\%$ nodes in the top percentile in terms of network speed being able to confirm the transaction before the bottom $10\%$ of the nodes receive it.\footnote{We make an additional assumption here that all nodes should have equal computing power. This is ideal from the expectation that the system should be decentralized. Hence, computational power should be ideally distributed equally among the nodes.}
\par Let $d$ be the advantage offered to the top $M$ fraction of the nodes in terms of time (in our example, this amounts to 11 seconds), $p$ is the probability of query being successful, and $H$ be the hash rate of the network.
    \begin{theorem}[Lower bound of $\{p_f\}^{M}_{m}$]
    \label{thm:frontrun}
    $\{p_f\}^{M}_{m} >  M\lambda d - \frac{1}{2}\big(M\lambda d\big)^2$
    \end{theorem}
\begin{proof}
We show that the probability of the event that $M$ fraction of the network manages to mine a block in time $d$ increases with the block creation rate $\lambda$. A formal proof is provided in Section \ref{proof:frontrun} of the supplementary material.
\end{proof}
Theorem \ref{thm:frontrun} shows that $\{p_f\}^{M}_{m}$ increases monotonically with increasing the block creation rate $\lambda$ since its lower bound increases monotonically. One can observe that in order to keep the probability of this event sufficiently low, $d$ must be reduced while increasing $\lambda$ to scale the blockchain. Although, it may seem that the lower bound is independent of the bottom $m$ percentile selected but this would have been incorporated in $d$ since $d$ increases as $m$ decreases.
\subsubsection{Practical Significance of Theorem \ref{thm:frontrun}\label{par:frontrun_result}} As of October 2020, $\{p_f\}^{0.5}_{0.9}$ approximates to $0.01$ for the Bitcoin Network. However, if we were to scale the Bitcoin Network to a throughput similar to offered by the likes of Visa Network by increasing the block creation rate to 566 blocks every 10 minutes (by reducing the difficulty and keeping the block size same) $\{p_f\}^{0.5}_{0.9}$ goes up to $0.99$ \footnote{We assume there are no significant improvements in the peer-to-peer overlay network}. Unfortunately, since at the time of writing, statistics on transaction propagation times in Ethereum were not published, we estimate $\{p_f\}^{0.5}_{0.9}$ for Ethereum by using the same delays as Bitcoin. We find that the estimated value of $\{p_f\}^{0.5}_{0.9}$ (as of October 2020) is around $0.36$, which is considerably higher than that of Bitcoin. In Figure \ref{fig:p_f}, we plot the variation of $\{p_f\}^{0.5}_{0.9}$ with increase in the throughput.
\setlength{\textfloatsep}{0.1cm}
\setlength{\floatsep}{0.1cm}
\begin{figure}
    \centering
      \includegraphics[width=0.5\textwidth, natwidth=264.53, natheight=141.3]{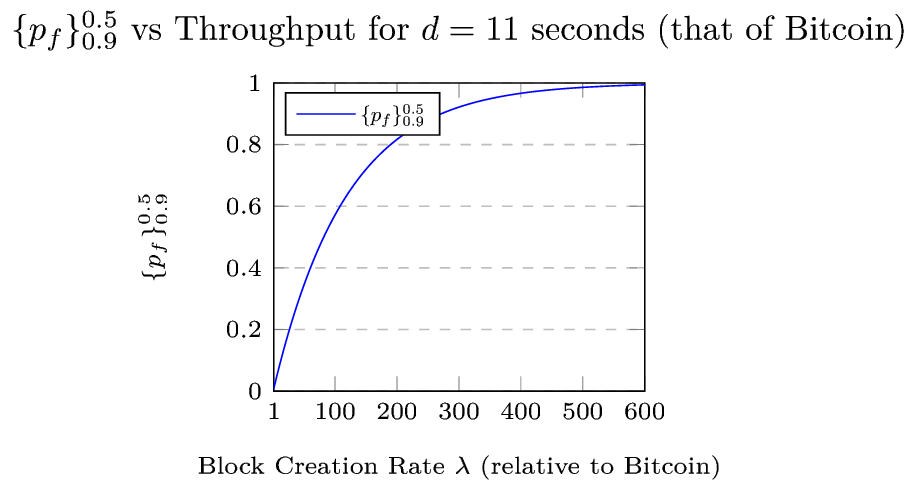}
      \caption{Variation of $\{p_f\}^{0.5}_{0.9}$ as we scale the blockchain}
      \label{fig:p_f}
\end{figure}
\subsubsection{Variation in Block Rewards\label{ssec:variation}}
Our analysis in Section \ref{ssec:frontrunning_analysis} shows that as we scale the blockchain to higher throughputs, some agents would be able to grab high-value transactions before others consistently. This means that some agents would be producing blocks with higher rewards, while others would produce blocks with lower rewards. Hence, a lack of network fairness may be able to induce a greater variation in block rewards. We discuss further implications of this variation in Section \ref{sec:strategic}.
\paragraph{Lack of incentive for information propagation} The inability of the peer-to-peer network to keep up with the ledger's desired throughput is further exacerbated by the fact that the nodes do not have any incentive to participate in broadcasting information. In fact, they have an incentive to keep the knowledge of transactions to themselves, as shown by Babaioff et al. \cite{babaioff2012bitcoin}. By broadcasting a transaction to other nodes, a agent is potentially increasing the number of nodes competing to include the transaction in their blocks and collect the corresponding transaction fees. Thus, offering additional incentive to agents for propagating transactions may speed up the the broadcast of a transaction.
\subsection{Publishing Fairness in Broadcasting Blocks}
\begin{figure}
\centering
\includegraphics[width=0.5\textwidth, natwidth=114.52, natheight=84.76]{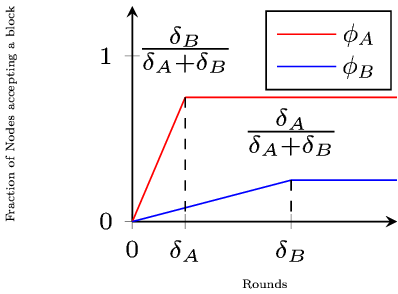}
\caption{Fraction of nodes accepting a block vs rounds}
\label{fig:propagation_plot}
\end{figure}
Let us assume that the execution happens in ``rounds'' in which the nodes make $q$ queries each to the Hash Function. At the boundaries of rounds, the nodes can communicate with their neighboring nodes. In our example here, we assume the duration of a round to be 1 second. If a node succeeds in mining a block in a round, it will begin broadcasting it to its neighbors when the round ends. Similarly, if a node receives a block at the beginning of a round, it will broadcast it to its neighbors at the end of that round.
\par We assume that there are $n$ nodes, and all nodes are honest. Hence, they follow the Bitcoin protocol's strategy of picking the oldest block in case of a tie and publishing a block as soon as it is mined. 
\par Consider the event in which exactly two nodes, A and B mine a block simultaneously in the same round. Let $\delta_A$ and $\delta_B$ ($\delta_A \leq \delta_B$, without loss of generality) be the delay associated with broadcasting the blocks to the entire network. At a time $t>\delta_B$, the network will be split into two fractions:
\begin{itemize}
    \item $\phi_A$: The fraction of nodes in the network which claim to have received the block mined by A before the block mined by B.
    \item $\phi_B$: The other fraction of nodes in the network which claim to have received the block mined by B before the block mined by A.
\end{itemize}
\par Since by $\delta_B$, all nodes in the network would have received the blocks mined by A and B, $\phi_A + \phi_B = 1$.
\par Let $\phi_A^i$ and $\phi_B^i$ be the fraction of network accepting A and B at the $i^\text{th}$ round. Then, $\alpha_f$ can be approximated by Theorem \ref{thm:alpha_f}.
\begin{figure*}[!t]
    \begin{minipage}{0.32\textwidth}
    \centering
        \includegraphics[width=1.05\textwidth, natwidth=348.94, natheight=304.51]{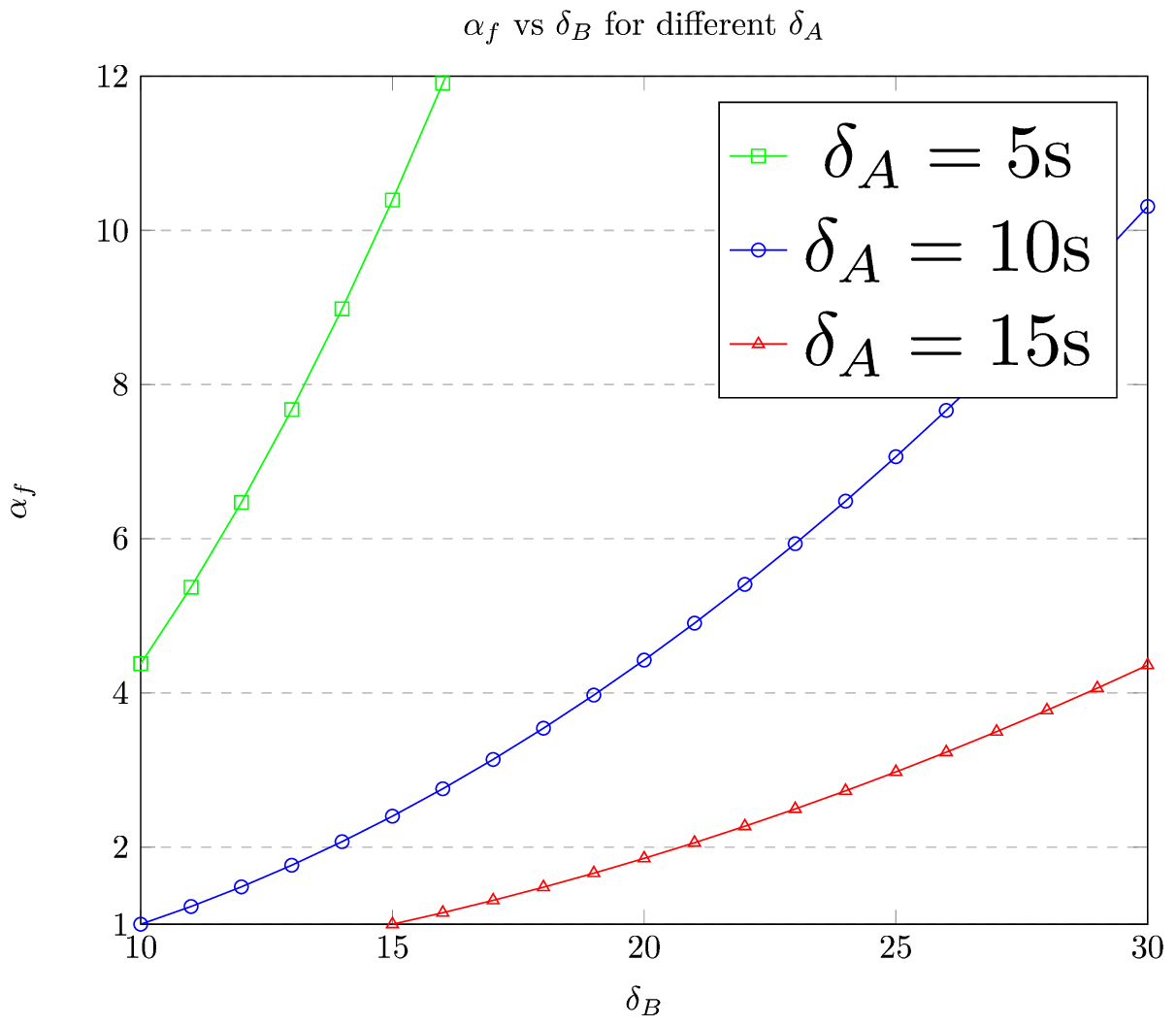}
        \label{fig:alpha_vs_delta}
\end{minipage}\hfill\begin{minipage}{0.32\textwidth}
\centering
        \includegraphics[width=1.05\textwidth, natwidth=348.94, natheight=304.51]{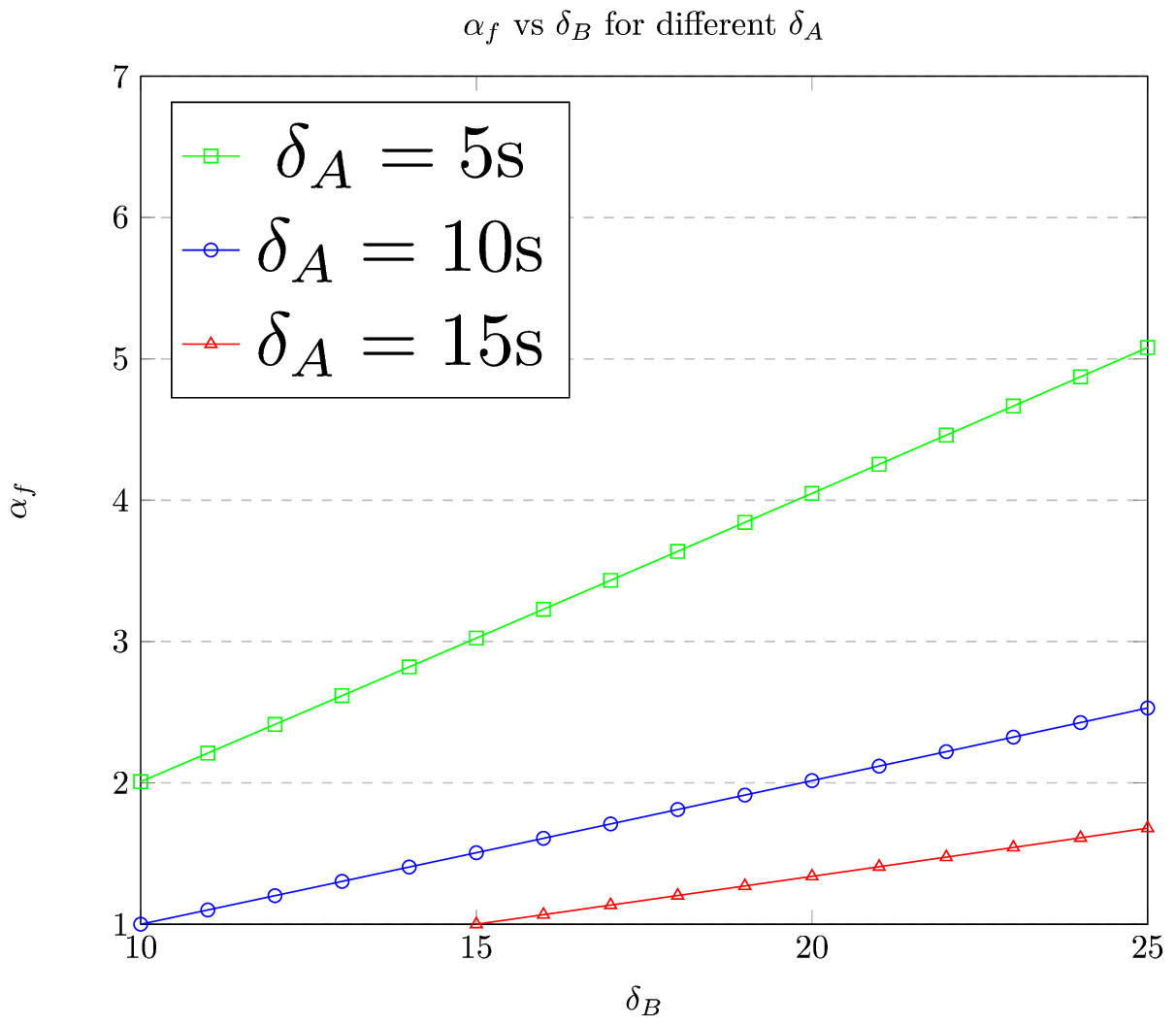}
        \label{fig:alpha_vs_delta2}
\end{minipage}\hfill\begin{minipage}{0.32\textwidth}
\centering
        \includegraphics[width=1.05\textwidth, natwidth=351.42, natheight=305.95]{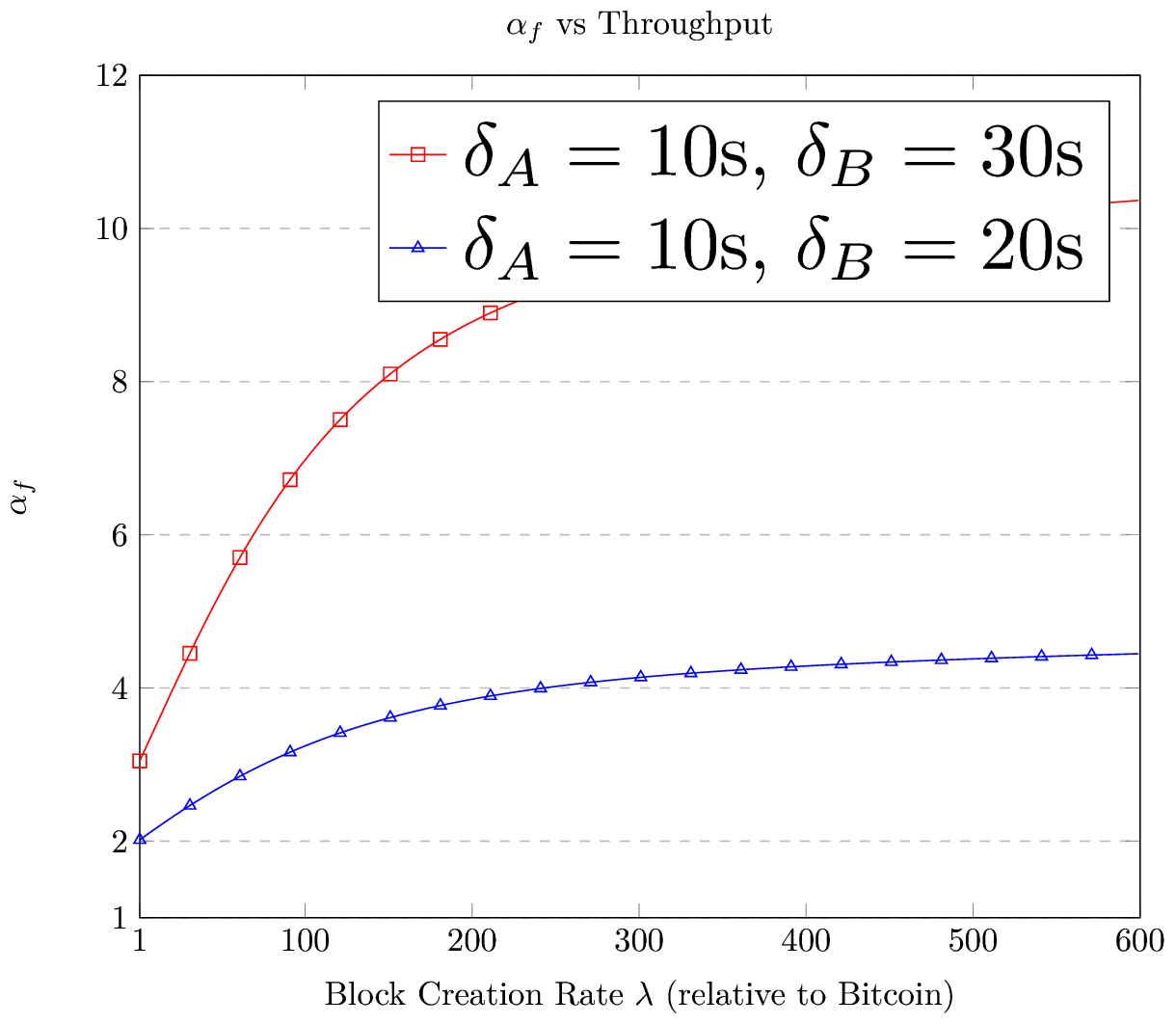}
        \label{fig:alpha_vs_f}
  \end{minipage}
\end{figure*}
\begin{theorem}[Approximation of $\alpha_f$]
\label{thm:alpha_f}
\begin{equation}
    \alpha_f = \frac{
    \psi
    }{
    1-\psi
    }
\end{equation}
where\\
\resizebox{0.99\hsize}{!}{$
    \psi = \sum_{i=1}^{\infty}  \bigg[\prod_{j=0}^{i-1} [(1-\text{fail}(\phi_A^j))(1-\text{fail}(\phi_B^j))+\text{fail}(\phi_A^j)\text{fail}(\phi_B^j)]
    \times (1-\text{fail}(\phi_A^i))\text{fail}(\phi_B^i) \bigg]
     $}
\end{theorem}
\begin{proof}
The proof of follows from first finding the probability of A being successful to getting its block accepted in $i^\text{th}$ round conditioned on the probability that neither of the blocks gain majority till the $(i-1)^\text{th}$ round. We then use \emph{Baye's Theorem} to find out the total probability. A formal proof is presented in Section \ref{proof:alpha_f} of the supplementary material.
\end{proof}
\subsubsection{Calculating $\alpha_f$}
 If we assume the propagation of the block in the network to be linear\footnote{Although, we expect it to be exponential in practice due to the GOSSIP algorithm, a linear assumption is rather optimistic and our results would be more pronounced in the exponential case.}, then by the end of $\delta_B$,  $\phi_A  = \frac{\delta_B}{\delta_A+\delta_B}$ and $\phi_B  = \frac{\delta_A}{\delta_A+\delta_B}$. Accordingly, the plot of $\phi_A$ and $\phi_B$ will be as shown in Figure \ref{fig:propagation_plot}. We then calculate $\alpha_f$ according to Theorem \ref{thm:alpha_f} and plot for varying network delays as well as varying throughputs.
\subsubsection{Results on $\alpha_f$ \label{par:a_f_results}}
\par In Figures \ref{fig:alpha_vs_delta} and \ref{fig:alpha_vs_delta2}, we plot the variation of $\alpha_f$ with the increase in $\delta_B$ for different values of $\delta_A$. We find that $\alpha_f$ grows exponentially with an increase in $\delta_B$ or a decrease in $\delta_A$, which implies that even small differences in network delays can make the system unfair.
\par In Figure \ref{fig:alpha_vs_f}, we plot the variation of $\alpha_f$ as we increase throughput for a fixed $\delta_A$ and $\delta_B$. We find that $\alpha_f$ grows with an increase in the block creation rate $\lambda$, which implies that as we scale the system, it becomes more unfair (from Remark \ref{remark:alpha_f_unfair}). However, the $\alpha_f$ increases slowly after a certain $\lambda$. We must note that the probability of simultaneously mining a block still keeps increasing exponentially, which is not factored in $\alpha_f$ (from Remark \ref{remark:alpha_f_simultaneous}). Hence, the overall ratio of blocks that A is able to include in the blockchain as compared to B still keeps increasing.
\subsection{A Few Tips and Tricks}
Improving the delays associated with broadcasting transactions and blocks for a node may not be as simple as improving the quality of a node's internet connection. A significant factor of the delay also depends upon the structure of the peer-to-peer network, the number of neighbors, and their delays. Decker and Wattenhofer \cite{decker2013information} suggest pipelining of information propagation to reduce the latency. However, we are more concerned with the changes a single node can implement to reduce its delay. Stathakopoulou et al. \cite{stathakopoulou2015faster} suggest implementing a Content Distribution Network that connects to nodes based on their geographic proximity. There have also been a few high-speed alternative ``relay'' networks developed to broadcast information quickly among a subset of nodes part of the network (\cite{bitcoin_fibre}, \cite{falconnet} and \cite{klarman2018bloxroute}). However, there are many concerns raised about the centralized nature of these networks. A node that is a part of such a network could undoubtedly gain an advantage over others. However, doing so may lead to further centralization, defeating the purpose of having a distributed ledger.
\section{Fairness and Strategic Deviations\label{sec:strategic}}
A lack of fairness causes some agents to gain more than their fair share of rewards, whereas some agents gain less than their fair share of rewards. Since all agents have similar underlying costs related to mining, it would make mining less profitable or even loss-making for some agents. We believe that this might open up the pandora's box of strategic deviations that might be not only unfair to the honest players but also detrimental to the health of the blockchain. Until now, we had considered that all players were honest and will not deviate from the protocol. However, the players may choose to deviate from the protocol if the deviation cannot be detected or penalized. When mining is not fair to some agents, they may have an incentive not to accept the consensus and depart from the honest strategy. This disagreement could possibly be reflected by forks that cause the agents in the system to split their votes. These deviations might be harmful to the health of a blockchain and make it less secure against adversaries. We now briefly discuss a few possibilities if they act rationally to maximize their expected reward.

\par Carlsten et al. \cite{carlsten2016instability} had shown that when there is a large variance in the reward earned from blocks, it might be profitable to intentionally fork blocks with high rewards. \emph{Petty mining}  \cite{carlsten2016instability} is a strategy in which, given a fork, the petty miner picks the fork, which has collected lesser transaction fees and, hence, offers the agent an opportunity to include transactions from the other fork and collect more transaction fee. \emph{Undercutting}  \cite{carlsten2016instability} is a strategy in which a agent intentionally forks a block with a high reward in order to collect some of the rewards while offering the rest of the reward to the petty agents that choose to mine on top of it. A slow node that does not have enough high-value transactions in its mempool might have an incentive to either fork the block mined by a frontrunner (undercutting) or given a fork pick the fork that offers an opportunity to collect a higher transaction fee (petty mining). If we assume that all agents are rational and hence petty miners since this strategy strictly dominates honest mining. Undercutting would allow a slower node to overcome both frontrunning and lack of publishing fairness. A slower node could fork a block mined by a faster node containing many high-value transactions due to frontrunning and include those transactions in its own block while leaving some of the transactions for others to include. Secondly, even if a node receives the block mined by a slower node later, it would drop the previous block and mine on top of this instead since it offers a higher reward.
\footnote{In fact, undercutting might even occur in case of an unintentional fork between a faster node and a slower node (since the block mined by a faster node will contain more high-value transactions). The probability of such forks increases along with $p_f$ and $\lambda$} In this case, it might not be in the best interest of the frontrunner to always pick the transactions offering higher rewards. 


\subsection{Modelling Bitcoin Mining As A Game}
We divide the network into two portions: \emph{slow} and \emph{fast}. The fast nodes can receive messages broadcasted by any node in the previous round, but the slow nodes have higher communication delays with certain nodes. Each node can choose from the following strategies:
\begin{enumerate}
    \item \texttt{petty}: The petty mining strategy described in \cite{carlsten2016instability}, it is same as the default strategy in case of no forks. It weakly dominates the default compliant strategy prescribed by Bitcoin but it is not harmful to the security of the blockchain on its own.
    \item \texttt{minor\_undercutting}: A node will undercut if the longest chain's reward is below a certain threshold. However, it would leave out a small constant reward as an incentive for the subsequent agents that pick the block.
    \item $\texttt{major\_undercutting}(\kappa)$: A node will undercut if the longest chain's reward is below a certain threshold. However, it would leave out a significant portion of the reward ($\kappa$) as an incentive for the subsequent agents that pick the block.
\end{enumerate}
The resulting game could be analyzed as a two-player bi-matrix game.
\subsection{Simulation-based Analysis}
\subsubsection{Simulator Details}
In order to study undercutting based strategic deviations in blockchains at high throughputs, we developed a lightweight simulator (described in Algorithm \ref{algo:simulator}).
\setlength{\textfloatsep}{0.1cm}
\setlength{\floatsep}{0.1cm}
\begin{small}
\begin{algorithm}[!tbp]
\SetAlgoLined
\KwInput{Distance matrix $D_{n\times n}$, Number of rounds $R$, Strategy $s_i \in S \forall i \in [n]$}
\KwResult{Expected reward of each player $P_i\quad \forall i \in [n]$}
\tcc*[h]{Set of mined blocks}\\
Initialize $B = {\text{genesis\_block}}$\\
Initialize $\texttt{longest\_chain} = 1$\\
 \For{$r=1$ \KwTo $R$}{
    \For{$i=1$ \KwTo $n$}{
        \If{$P_i$ wins a lottery in round $R$} {
            $B' = \emptyset$\\
            $\texttt{max\_reward} = 0$\\
            \For{$b \in B$} {
                \If{$b$ satisfies strategy $s_i$ and $r - b.\texttt{round} \geq D[b.\texttt{miner}][i]$}{
                    $B' = B'\cup b$\\
                    $\texttt{max\_reward} = \max(r - b.\texttt{round} + b.\texttt{leftover}, \texttt{max\_reward})$\\
                }
            }
            \For{$b \in B'$}{
                \If{$r - b.\texttt{round} + b.\texttt{leftover} \geq \texttt{max\_reward}$}{
                    $\hat{b}.\texttt{parent} = b$\\
                    $\hat{b}.\texttt{miner} = i$\\
                    $\hat{b}.\texttt{height} = b.\texttt{height} + 1$\\
                    $\hat{b}.\texttt{reward} = \min(r - b.\texttt{round} + b.\texttt{leftover}, \texttt{max\_block\_size})$\\
                    \uIf{$s_i \in \texttt{major\_undercut}$}{
                        $\hat{b}.\texttt{leftover} = \kappa$\\
                        $\hat{b}.\texttt{reward} = \hat{b}.\texttt{reward} - \kappa$\\
                    }\uElseIf{$s_i \in \texttt{minor\_undercut}$}{
                        \tcc*[h]{A small quantity $d$}\\
                        $\hat{b}.\texttt{leftover} = d$\\
                        $\hat{b}.\texttt{reward} = \hat{b}.\texttt{reward} - d$\\
                    }\Else{
                        $\hat{b}.\texttt{leftover} = 0$
                    }
                    $B = B\cup\hat{b}$\\
                    $\texttt{longest\_chain} = \max(\hat{b}.\texttt{height}, \texttt{longest\_chain})$\\
                }
            }
        }
    }
 }
 \KwRet{Expected reward of each player averaged over all chains with length = $\texttt{longest\_chain}$}
 \caption{\label{algo:simulator}Simulator}
\end{algorithm}
\end{small}
We tested the following strategies: \\\begin{inparaenum}[(a)]
\item $\texttt{major\_undercutting}(1.5)$ ($S_1$),
\item $\texttt{major\_undercutting}(1)$ ($S_2$),
\item \texttt{minor\_undercutting} ($S_3$), and
\item \texttt{petty} ($S_4$).
\end{inparaenum}
We assigned different strategies to slow and fast nodes to produce the payoff matrix in Table \ref{table:payoff}\footnote{Due to computational constraints, we produced results for $\lambda = 300$ times that of Bitcoin which would yield a throughput 60\% that of the Visa Network but we expect the results to be even more pronounced as we scale the blockchain further.}.
\begin{table}[!tbp]
\centering
    \begin{tabular}{rllll}
    \hline
         & $S_1$& $S_2$ & $S_3$ & $S_4$
         \\
         \midrule
         $S_1$ & (75.22, 19.13)	& (71.72, 23.12) & (74.29, 21.54) &	{\color{blue} (73.75, 22.17)} \\
         $S_2$ & (75.22, 19.86)	& (76.05, 19.56) & (72.17, 24.24) &	{\color{blue} (74.99, 21.74)} \\
         $S_3$ & {\color{blue} (63.30, 33.17)} & {\color{blue} (63.35, 33.84)} & {\color{blue} (66.90, 30.68)} & {\color{blue} (74.02, 23.6)} \\
         $S_4$ & {\color{blue} (63.41, 33.06)} & {\color{blue} (64.04, 33.29)} & {\color{blue} (56.66, 40.78)} & {\color{blue} (67.54, 30.26)} \\
         \hline
    \end{tabular}
    \caption{Payoff Matrix averaged over 100 simulations}
    \label{table:payoff}
    {\raggedright The utilities shown here are the percentage of total reward grabbed by the fast and slow sets collectively. (They do not add up to 100\% since some blocks may be underutilized)\par}
    {\raggedright The strategies which would be removed by iterated removal of dominant strategies (with tolerance of $\pm 1)$ are shown in {\color{blue}{blue}}.}
\end{table}
\subsubsection{Equilibrium Analysis}
\begin{definition}[Weakly Dominated Strategy]{\cite{osborne1994course}}
A strategy $s_i \in S_i$ is said to be weakly dominated if $\exists\ s_i' \in S_i$ for an agent $i$ if\\
\resizebox{0.99\hsize}{!}{
$
u_i(s_i', s_{-i}) \geq u_i(s_i, s_{-i})\ \forall s_{-i} \in S_{-i}\ \text{and}\ u_i(s_i', s_{-i}) > u_i(s_i, s_{-i})\ \text{for some}\ s_{-i} \in S_{-i}
$}\\
where $u_i(s_i, s_{-i})$ is the payoff for agent $i$ if he/she chooses the strategy $s_i$ and the other agents' vector of strategies is $s_{-i}$
\end{definition}
\begin{definition}[Mixed Strategy Nash Equilibrium (MSNE)]{\cite{osborne1994course}}
A strategy profile $(\sigma_1^*, \sigma_2^*, \ldots ,\sigma_n^*)$ is called a Mixed Strategy Nash Equilibrium (MSNE) for $n$ agents, if for each agent $i$, $\sigma_i^*$ is the best response to $\sigma_{-i}^*$. That is, $\forall i \in n$,
\[
u_i(\sigma_i^*, \sigma_{-i}^*)  \geq u_i(\sigma_i, \sigma_{-i}^*), \forall \sigma_i \in \Delta(S_i)
\]
where $\sigma_i = (p_1^i, p_2^i, \ldots, p_k^i)$ is the mixed strategy played by the player $i$ in which he/she chooses strategy $s_j$ with probability $p_j^i$.
\end{definition}
By applying \emph{Iterated Removal of Dominated Strategies} on the Payoff Matrix, we discard strategies $S_3$ and $S_4$ from the solution set of fast agents and discard $S_4$ from the solution set of slow agents. We then find the following two mixed strategy nash equilibria among the remaining set of strategies:
\begin{enumerate}
    \item Fast agents choose $S_1$ with probability of $0.74$ and $S_2$ rest of the time, while slow agents choose $S_2$ with probability of $0.32$ and $S_3$ rest of the time.
    \item Fast agents choose $S_1$ with probability of $0.06$ and $S_2$ rest of the time, while slow agents always pick $S_1$.
\end{enumerate}
We make the following additional remarks based on the payoff matrix:
\begin{enumerate}
    \item The blockchain system would have been secure against any strategic deviations if $(\texttt{petty}, \texttt{petty})$ had been an equilibrium strategy. However, we observe that it is dominated by nearly all undercutting strategies for slower nodes. Hence, it would always be more profitable for the slower nodes to undercut. A lack of publishing fairness could explain this.
    \item If the slow nodes choose \texttt{minor\_undercutting}, it would be more profitable for the faster nodes to choose\\ \texttt{major\_undercutting}.
    \item If all the nodes choose \texttt{major\_undercutting} the total revenue gathered by the network reduces to roughly 94.3\% from 97.8\% indicating that fewer transactions are being added to the longest chain. This means that the throughput is being under-utilized.
    \item The equilibria strategies are even worse for the slower nodes in terms of fairness since they grab an even smaller share of reward as compared to the strategy where all players act honestly.
\end{enumerate}
\par Thus, if agents act rationally not only would security of the blockchain be adversely affected, the lack of fairness among the rewards received by the slower miners would be exacerbated.

\section{Related and Future Work}
Garay et al.\cite{cryptoeprint:2014:765} and Kiyayas et al. \cite{cryptoeprint:2015:1019} show that existing blockchain protocols suffer from a loss of security properties as we scale the system. These security properties are fundamental to the operation of a robust DL. \cite{cryptoeprint:2014:765} and \cite{cryptoeprint:2015:1019} consider a model with two types of players, honest and adversarial. They consider that the adversary tries to attack the ledger by forking the blockchain. A successful fork would allow the adversary to perform a double-spending attack. We, on the other hand, analyze the issues that may arise even if all players are honest. We show that as we scale the system, not only do the security properties suffer, but fairness also suffers, and hence, our work is parallel to theirs. 
\par Further, there have been many novel blockchain protocols proposed to scale blockchain without losing the security properties, e.g., OHIE \cite{ohie}, IOTA \cite{popov2018tangle}, GHOST \cite{cryptoeprint:2013:881}, GhostDAG/Phantom \cite{cryptoeprint:2018:104}, and many more coming up every day. We believe that our fairness measures are generalized enough also to be applied to them. However, the analysis may vary according to the design of the protocol. Even though many of the DAG-based blockchain protocols manage to solve the issue of publishing fairness by allowing ``off-chain'' blocks to be considered as a part of the blockchain, the frontrunning issue persists. Due to unfairness in frontrunning, we find that they suffer from a different style of undercutting attacks. We demonstrate an example of such attack in a blockchain protocol known as OHIE in Section \ref{ohie_flaw} of the supplementary material.

\par Researchers have been analyzing blockchain systems from a lens of game theory, Chen et al. \cite{chen2020nonlinear} use game-theory to establish the tradeoff between full verification, scalability, and finality-duration. However, they do not consider network delays in their model, and hence, their bounds are quite optimistic. \cite{rational2020} study the consensus in a setting with rational as well as honest agents. \cite{www19} analyze Bitcoin's transaction fees using auctions. \cite{bitcoinf} study fairness in the transaction fees collected by the nodes in a transaction fee only model and propose a fair transaction processing protocol.

\par \cite{pass2017fruitchains} define $\eta$-approximate fairness and propose a blockchain protocol that satisfies $\eta$-approximate fairness in the case where players have symmetric network access, whereas we study the case where players have asymmetric network access.
\section{Conclusion}
In this paper, we introduced a notion of \emph{network fairness}, investigated the factors that influence network fairness, and studied its impact on the agent' revenue. We assumed a model in which players have asymmetric network connections, i.e., some players may have a faster connection while others may have a slower connection and described two mechanisms via which this asymmetry could result in a loss of fairness in the system. We considered two events associated with the mining process, frontrunning and block publishing, and measured fairness for these events. We showed that fairness can be quantified via $p_f$ and $\alpha_f$. $p_f$ is a measure of \emph{frontrunning} due to network delays in broadcasting a transaction. $\alpha_f$ or publishing fairness is the ratio of blocks of one node that get accepted due to network delays over blocks of another node. We found that both of them deteriorate as we increase the throughput of existing protocols. Hence, even though it might look like the agents walk together (or the system is fair) while the throughput is low, at higher speeds, some agents may run faster (or gain more than their fair share of rewards).
\par We also discussed that not only does a lack of fairness impacts the revenue of some agents, it might also create an incentive for them to deviate from the honest mining strategy, which might impact the security of the blockchain system and further exacerbate lack of fairness in rewards.
\par Thus, we conclude that even though blockchain is an ambitious technology, its potential is hindered by the underlying network infrastructure.



\bibliographystyle{ACM-Reference-Format} 
\bibliography{references}

\clearpage
\appendix
\section{Fail Function}
\label{proof:fail}
Let us denote $p$ to be the probability of a nonce being successful and $q$ to be the rate of queries to the hash function. In time $t$, the miner will be able to make $q\times t$ such queries.\\\\
The probability of a query being unsuccessful will be $(1-p)$.\\\\
$\implies$ the probability of $q\times t$ such queries being unsuccessful will be $(1-p)^{q\times t}$.\\\\
For $n$ players (the total number of players in the system), the number of queries in the same duration will be $n\times q\times t$\\\\
$\implies$ the probability of these nodes failing to mine a block will be $(1-p)^{n\times q\times t}$.\\\\
The value $n\times q$ is known as the hash rate of the system, denoted by $H$.\\\\
Let us denote the probability of $\phi$ fraction of the network failing to mine a block by the function $\text{fail}(\phi,t) = (1-p)^{\phi\times nqt} = (1-p)^{\phi\times Ht}$.\\\\
For convenience, we sometimes use $\text{fail}(\phi) = \text{fail}(\phi, 1)$, the probability of $\phi$ fraction of the network failing to mine a block in a unit of time.
\section{Proof of Theorem \ref{thm:frontrun}}
\label{proof:frontrun}
\begin{theorem*}[Lower bound of $\{p_f\}^{M}_{m}$]
$\{p_f\}^{M}_{m} >  M\lambda d - \frac{1}{2}\bigg(M\lambda d\bigg)^2$
\end{theorem*}
\begin{proof}
Let $\texttt{frontrun\_2}(M,m)$ denote the event that top $M$ fraction of nodes succeed in mining the block before the transaction reaches the bottom $1-m$ fraction.
\begin{align*}
       &P(\texttt{frontrun\_2}(M,m)) \\&= 1 - P(M\ \text{fraction of nodes fail to mine a block in time}\ d)\\
    &=\ 1- \text{fail}(M, d)\\ &= 1-(1-p)^{M\times Hd} \\&> pMHd - \frac{1}{2}\bigg(pMHd\bigg)^2 \\&= M\lambda d - \frac{1}{2}\bigg(M\lambda d\bigg)^2
\end{align*}
\end{proof}
\section{Proof of Theorem \ref{thm:alpha_f}}
\label{proof:alpha_f}
\begin{theorem*}
\begin{equation}
    \alpha_f = \frac{
    \psi
    }{
    1-\psi
    }
\end{equation}
where
\begin{equation}
\begin{split}
    \psi = \sum_{i=1}^{\infty}  \bigg[&\prod_{j=0}^{i-1} [(1-\text{fail}(\phi_A^j))(1-\text{fail}(\phi_B^j))+\text{fail}(\phi_A^j)\text{fail}(\phi_B^j)] \\
    &\times (1-\text{fail}(\phi_A^i))\text{fail}(\phi_B^i) \bigg]
\end{split}
\end{equation}
\end{theorem*}
\begin{proof}
\par The proof follows from first finding the probability of A being successful in getting its block accepted in $i^\text{th}$ round conditioned on the probability that neither of the blocks gain majority till the $(i-1)^\text{th}$ round. We then use \emph{Bayes' Theorem} to find out the total probability.
\par The block mined by A will get accepted if the chain that includes A's block becomes longer than the one that includes B. As soon as the longer chain is received by a node that had previously accepted B, the node must reject B and accept A instead. We assume that as soon as the chain mined by either fraction becomes longer than that of their counterpart, their counterpart will switch to the longer chain. Let $\phi_A^i$ and $\phi_B^i$ be the fraction of network accepting A and B at the $i^\text{th}$ round. We slightly abuse notation here to use A and B to refer to the blocks mined by A and B, respectively.
\begin{equation}
    \begin{split}
        &P\Big(\frac{\text{A gets accepted in $i^\text{th}$ round}}{\text{The network is undecided in $(i-1)^\text{th}$ round}}\Big) 
        \\&= (1-\text{fail}(\phi_A^i))\times\text{fail}(\phi_B^i)
    \end{split}
\end{equation}
The network fails to decide in the $i^\text{th}$ round if the length of the chains remain equal, i.e.,  either both fractions mine a block (which is highly unlikely), or both fractions fail to mine a block.
\begin{equation}
    \begin{split}
        &P\Big(\frac{\text{The network fails to decide in $i^\text{th}$ round}}{\text{The network is undecided in $(i-1)^\text{th}$ round}}\Big) \\&= (1-\text{fail}(\phi_A^i))\times(1-\text{fail}(\phi_B^i))+\text{fail}(\phi_A^i)\times\text{fail}(\phi_B^i)
    \end{split}
\end{equation}
\begin{align*}
\begin{split}
&P(\text{The network is undecided in $i^\text{th}$ round})\\
&= [(1-\text{fail}(\phi_A^i))(1-\text{fail}(\phi_B^i))+\text{fail}(\phi_A^i)\text{fail}(\phi_B^i)]\\
&\qquad\times P(\text{The network is undecided in $(i-1)^\text{th}$ round})\\
\end{split}\\
&= \prod_{j=0}^{i} (1-\text{fail}(\phi_A^j))(1-\text{fail}(\phi_B^j))+\text{fail}(\phi_A^j)\text{fail}(\phi_B^j)\\
\begin{split}
    &P(\text{A eventually gets accepted}) = \psi\\
    &= \sum_{i=1}^{\infty}P(\text{The network is undecided in $(i-1)^\text{th}$ round}\cap\\
    &\qquad\text{A gets accepted in $i^\text{th}$ round})\\
\end{split}\\
\begin{split}
    &= \sum_{i=1}^{\infty} \bigg[ P(\text{The network is undecided in $(i-1)^\text{th}$ round})\\
    &\qquad\qquad\times P\Big(\frac{\text{A gets accepted in $i^\text{th}$ round}}{\text{The network is undecided in $(i-1)^\text{th}$ round}}\Big) \bigg]\\
\end{split}\\
\begin{split}
    &= \sum_{i=1}^{\infty}  \bigg[\prod_{j=0}^{i-1}[(1-\text{fail}(\phi_A^j))(1-\text{fail}(\phi_B^j))+\text{fail}(\phi_A^j)\text{fail}(\phi_B^j)] \\
    &\qquad \times (1-\text{fail}(\phi_A^i))\text{fail}(\phi_B^i) \bigg]\\
\end{split}\\
    \alpha_f &= \frac{P(\text{A eventually gets accepted})}{P(\text{B eventually gets accepted})}\\
    \alpha_f &= \frac{P(\text{A eventually gets accepted})}{1-P(\text{A eventually gets accepted})}\\
    \therefore{}\alpha_f &= \frac{\psi}{1-\psi}\\
\end{align*}
\end{proof}
\section{Frontrunning in OHIE\label{ohie_flaw}}
\par In this section, we describe a strategic deviation for the OHIE Protocol based on frontrunning. OHIE is a permissionless blockchain protocol that aims to achieve high throughput while tolerating up to 50\% of the computational power being controlled by the adversary. 

\par OHIE composes $k$ (e.g., $k=1000$) parallel instances or ``chains'' of Nakamoto consensus. Each chain has a distinct genesis block, and the chains have ids from $0$ to $k-1$ (which can come from the lexicographic order of all the genesis blocks). For each chain, Bitcoin's longest-chain-rule is followed. The miners in OHIE extend the $k$ chains concurrently. They are forced to split their computational power across all chains evenly.

\par The total block ordering in OHIE is generated according to the increasing order of ranks and breaking ties by picking one with a lower chain id earlier. The miner of a block picks the rank of the block that follows it in a chain, i.e., the $next\_rank$. Notice that the chains may not be equal in length and might have different $next\_rank$s at their last positions. This implies that if a block is mined in a chain having lower $next\_rank$ than another chain, the block might end up earlier in the Total Block Ordering than a block that has already been mined.
\par According to the specifications of the protocol, a miner should pick the highest possible $next\_rank$ in order to ensure that the chain the block becomes a part of, catches up to the longest chain. A miner can also choose which blocks to mine on top of. If a rational agent wishes to insert his block earlier in the total ordering, he can choose a block that has picked a lower $next\_rank$ over one with a higher $next\_rank$. We describe a rational deviation based on this fact.

\par This deviation is similar to undercutting in Bitcoin, described by Carlsten et al. \cite{carlsten2016instability} We assume that at least some agents are rational and follow a \emph{petty compliant} strategy. The agents still choose to mine on top of the longest chains, but in case of a tie or fork, they pick the block offering lower $next\_rank$. Doing so provides an opportunity for \emph{strategic frontrunning} by possibly achieving a lower rank (and hence include the high-value transactions of already published blocks of higher rank). 

\subsection{Expected Reward of Including Transactions}
In our case, the frontrunning is not guaranteed to be successful. The block could end up on the chain from which the transaction was included. In which case, it will end up losing the transaction fees since it will not precede the original block from which the transaction was included. The block is equally likely to become a part of the $k$ chains. Therefore we define the expected reward of including a transaction as follows:
\[\mathbb{E}[R] = P(\text{Successful frontrunning})\times R\]
The probability of the succeeding to frontrun a block $b$ will be:
\begin{equation}
    \begin{split}
        P(\text{Frontrun}\ b) = \sum_{i=0}^{k-1} P(&\text{Block mining on the $i^\text{th}$ chain} \\&\cap \text{Block preceding block $b$ in TBO})
    \end{split}
\end{equation}
We can expect a rational agent to include the transactions from the mempool as well as transactions from other blocks that offer the highest expected reward.

Similar to \cite{carlsten2016instability} we find that for a \emph{petty compliant} miner \emph{frontrunning} is a better strategy since it guarantees as much reward as the honest node's strategy. As we try to scale the system to higher throuputs by increasing the number of chains, the number of possible blocks to frontrun will also increase (due to a higher probability of some chains being longer in length). Thus, the probability of \emph{strategic frontrunning} also increases as we try to scale the system.

\subsection{Undercutting Agents\label{ohie_uc}}
Now consider if a more aggressive agent does not mind forking a chain. In the following example (modified version from the original paper \cite{ohie}):
\begin{figure}[h]
    \centering
    \label{fig:ohie_initial}
    \includegraphics[width=0.5\textwidth, natwidth=874.653, natheight=415.383]{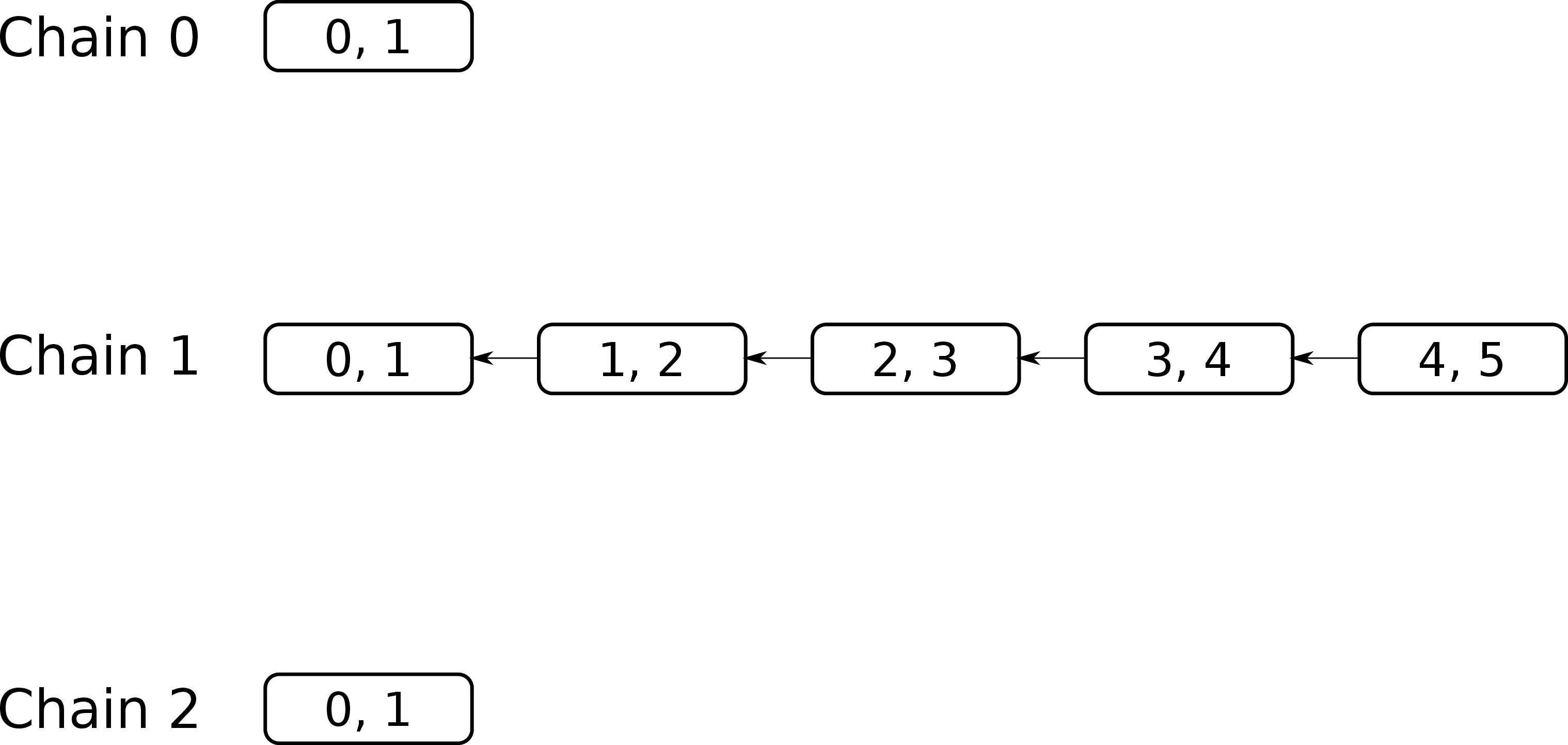}
    \caption{Example initial state of an OHIE execution with $k=3$. Each block is marked with a tuple $(rank, next\_rank)$.}
\end{figure}
\par In the example initial state, the chains are of unequal length. (Such an event is possible since the blocks extend chains at random, some chains can receive more blocks than others)
\begin{figure}[h]
    \centering
    \label{fig:ohie_honest}
    \includegraphics[width=0.5\textwidth, natwidth=1077.17, natheight=425.197]{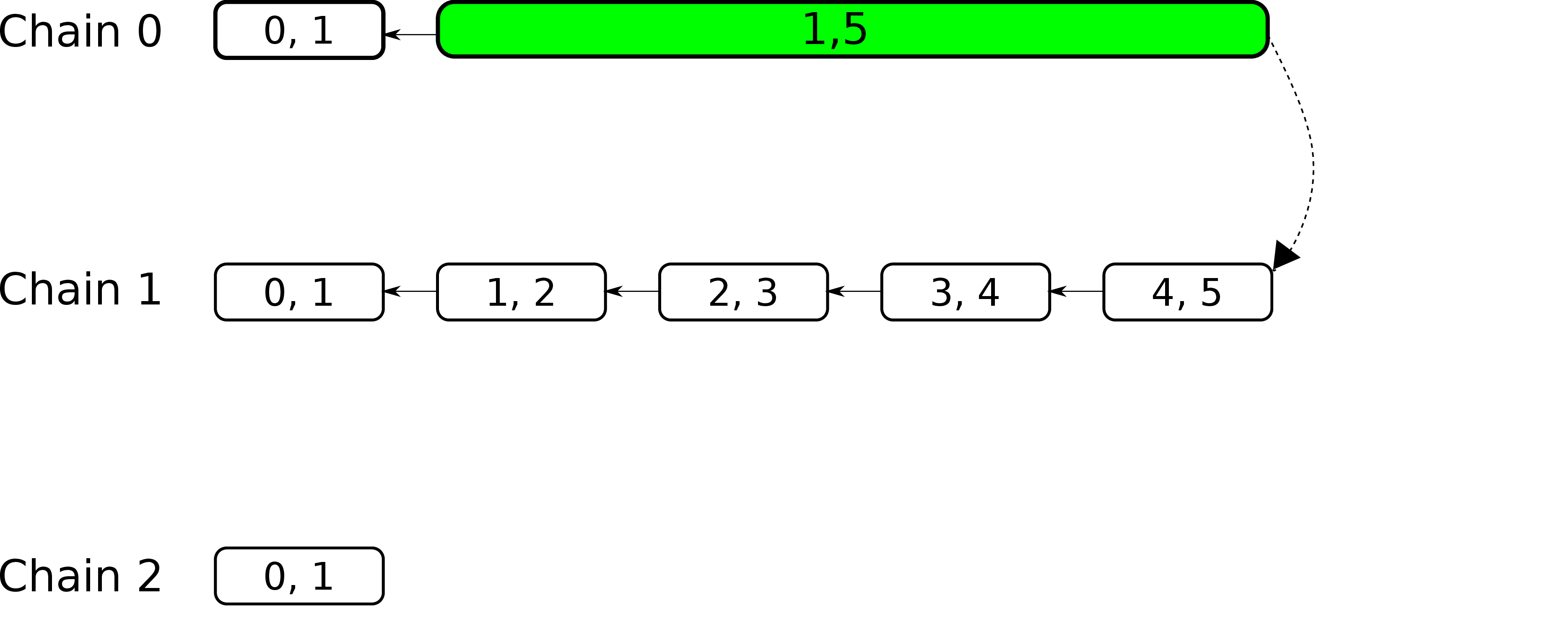}
    \caption{The state after the honest node extends Chain 0. A dotted arrow denotes the trailing pointer.}
\end{figure}
\par Let us say that an honest node mines the next block on Chain 0. Since the node follows the default strategy of setting the $next\_rank$ to be the maximum $next\_rank$ among all chains, it sets the $next\_rank$ to be 5. The mechanism by which the $next\_rank$ is specified is by including a trailing pointer to the last block on Chain 1. Hence, the $next\_rank$ of this block is implicitly set to the $next\_rank$ of the block that the trailing pointer points to.
\begin{figure}[h]
    \centering
    \label{fig:ohie_uc_mining}
    \includegraphics[width=0.5\textwidth, natwidth=1066.79, natheight=415.836]{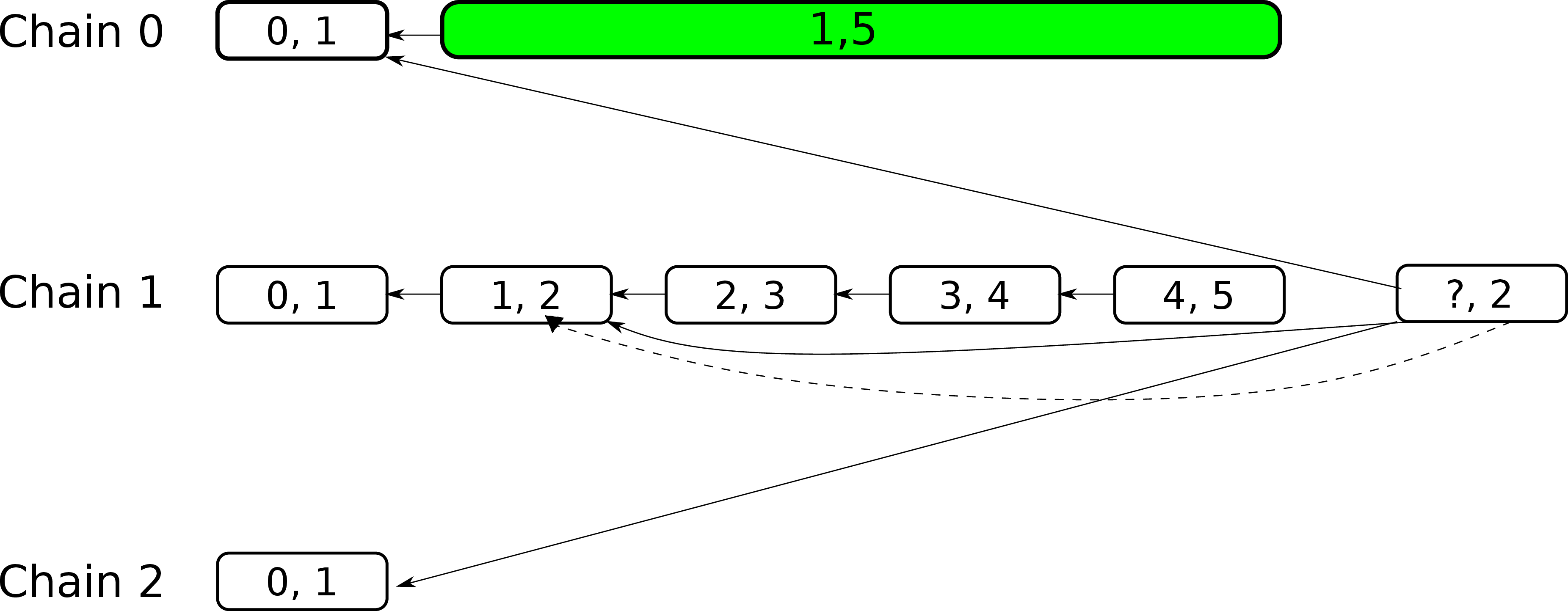}
    \caption{The block formulated by an aggressive miner for undercutting. The dotted arrow denotes the trailing pointer. The three solid arrows point to the Merkle tree of pointers to preceding blocks.}
\end{figure}
\par Consider a agent that tries to undercut aggressively. It picks the blocks it wants to drop (in this case $(1,5)$ on Chain 0 and $(2,3)$, $(3,4)$, and $(4,5)$ on Chain 1) and then picks the trailing pointer to $(1,2)$. However, it could have picked $(0,1)$ on any chain as the trailing pointer, in which case it would have been assigned a $next\_rank$ of $rank+1$. This deviation would have easily been detected since its trailing pointer would have lesser $next\_rank$ than the block it precedes, indicating the deviation. In our case, the agent tries to deviate in a manner that is not distinguishable from a fork. In order to do so, it selects a $next\_rank$ that is greatest among the blocks in its Merkle tree.
\begin{figure}[h]
    \centering
    \label{fig:ohie_uc_cases}
    \includegraphics[width=0.5\textwidth, natwidth=1077.17, natheight=425.197]{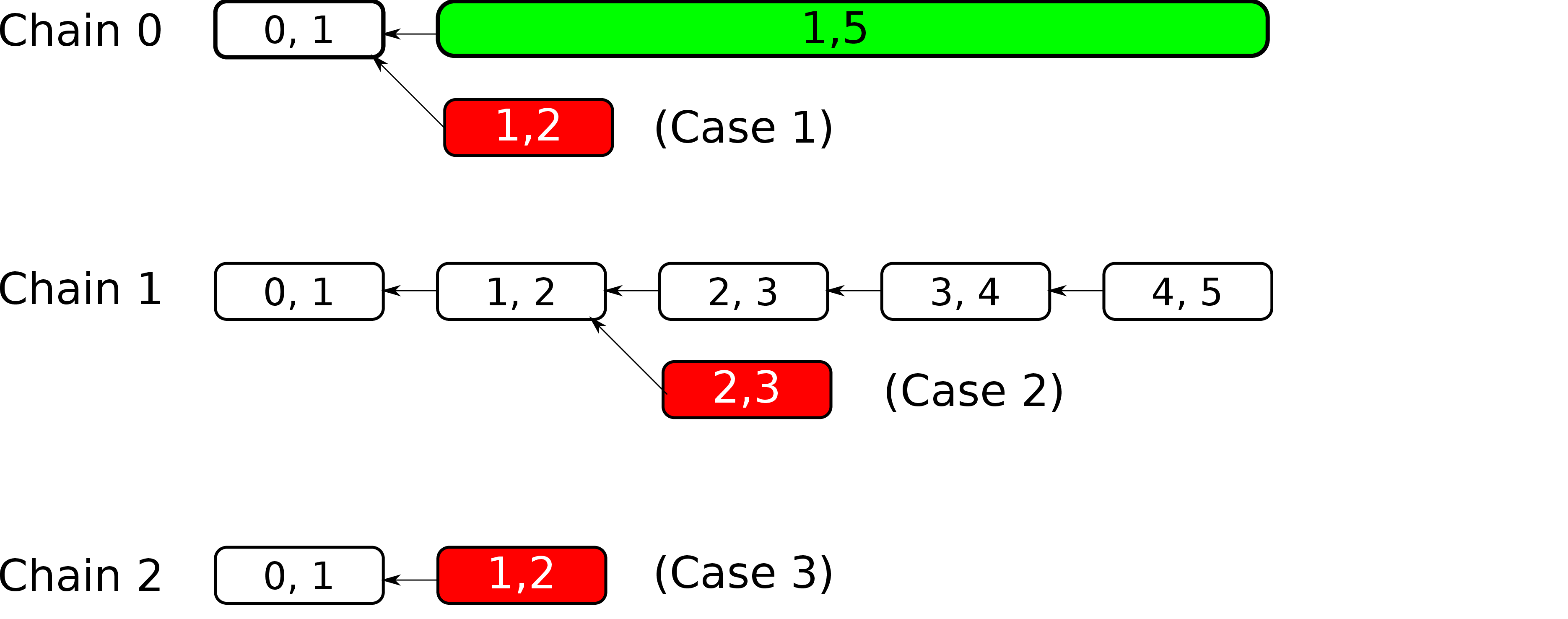}
    \caption{The three possible cases that could arise if the undercutter's mining is successful}
\end{figure}
\par If the undercutter is able to mine the block successfully, it may end up on one of the three chains depending upon the last $\log_2k$ bits of the hash. We consider these three cases separately:
\begin{itemize}
    \item If the block ends up on Chain 0, the new block forks the chain. Since the two forks are equal it is upto the nodes in the network to choose the fork they wish to extend. Choosing the undercutter's block in this case would be a better option for the \emph{petty compliant} agents since it offers a lower $next\_rank$. If the majority of the agents are \emph{petty compliant}, then the undercutter is successful.
    \item If the block ends up on Chain 1, the new block forks the chain. Since the undercutter's fork is shorter than the original chain, it would be orphaned by all agents. In this case, the undercutter is unsuccessful.
    \item If the block ends up on Chain 2, the new block extends the original chain. All agents will prefer to mine on top of undercutter's block. In this case, the undercutter is successful.
\end{itemize}
\par Hence, in 2 out of 3 cases, i.e., with a probability of $2/3$, the undercutter is successful. Therefore, the agent may undercut if the reward obtained by picking transactions from the blocks it drops is $3/2$ times more than the reward by mining honestly. That is, if the $\mathbb{E}[\text{Reward obtained by undercutting}] > \text{Reward obtained by mining honestly}$ then undercutting may be a better strategy.

\end{document}